\documentclass[a4paper,12pt]{article}
\usepackage{eurosym}
\usepackage[dvips]{graphicx}
\usepackage{amsfonts}
\usepackage{mathrsfs}
\usepackage{graphicx} 
\usepackage{epsfig}
\usepackage{amsmath, amsthm}
\usepackage{amssymb}

\newcommand{\ba}{\begin{array}}
\newcommand{\ea}{\end{array}}

\newcommand{\beq}{\begin{equation}}
\newcommand{\eeq}{\end{equation}}
\newcommand{\ben}{\begin{enumerate}}
\newcommand{\een}{\end{enumerate}}
\newcommand{\bit}{\begin{itemize}}
\newcommand{\eit}{\end{itemize}}

\newtheorem{prop}{Proposition}
\newtheorem{theorem}{Theorem}
\begin{document}
\title{Geometric formulation and multi-dark soliton solution to the
defocusing
complex short pulse equation}
\author{Bao-Feng Feng \\
\textsl{School of Mathematical and Statistical Sciences} \\
\textsl{The University of Texas Rio Grande Valley}
\and Ken-ichi Maruno \\
\textsl{Department of Applied Mathematics} \\
\textsl{Waseda University, Shinjuku-ku, Tokyo 169-8555, Japan} \\
\and Yasuhiro Ohta \\
\textsl{Department of  Mathematics} \\
\textsl{Kobe University, Rokko, Kobe 657-8501, Japan
}}
\maketitle

\begin{abstract}
In the present paper, we study the defocusing complex short pulse (CSP) equations both geometrically and algebraically. From the
geometric point of view, we establish a link of the complex coupled
 dispersionless (CCD) system with the motion of space curves in
 Minkowski space $\mathbf{R}^{2,1}$, then with the defocusing CSP equation via a
 hodograph (reciprocal) transformation, the Lax pair is constructed
 naturally for the defocusing CSP equation. We also show that
 the CCD system of both the focusing and defocusing types can be derived
 from
 the fundamental forms of surfaces such that their curve flows are
 formulated.
 In the second part of the paper, we derive the the defocusing CSP
 equation from the single-component extended KP hierarchy by the reduction
 method. As a by-product, the $N$-dark soliton solution for the defocusing CSP equation in the form of determinants for these equations is provided.
\end{abstract}
\section{Introduction}
Recently, a complex short pulse (CSP)
equation~\cite{Feng_ComplexSPE,KYKPRE14}
\begin{equation}
q_{xt}+q+\frac{1}{2}\left(|q|^{2}q_{x}\right) _{x}=0\,,  \label{CSP}
\end{equation}%
was proposed as an improvement for the short pulse (SP) equation
\begin{equation}
u_{xt}=u+\frac{1}{6}\left( u^{3}\right) _{xx}\,,  \label{SPE}
\end{equation}%
proposed by Sch\"{a}fer and Wayne~\cite{SPE_Org} to describe the propagation of ultra-short optical pulses in nonlinear media.
In contrast with a real-valued function for $u=u(x,t)$ in Eq.(\ref{SPE}), $q=q(x,t)$ in Eq.(\ref{CSP}) is a complex-valued function. Since the complex-valued function can contain the information of both amplitude and phase, it
is more appropriate for the description of the optical waves~\cite{Yarivbook}.
It was shown that the CSP equation (\ref{CSP}) is integrable in the sense that it admits a Lax pair and multi-soliton solutions~\cite{Feng_ComplexSPE,FengShen_ComplexSPE,FLZPhysD,FMO-PJMI}
In contrast with no physical interpretation for one-soliton solution
(loop soliton) to the SP equation~\cite{Sakovich2,Matsuno_SPE},
the one-soliton solution for the CSP equation is an envelope soliton with a few
optical cycles~\cite{Feng_ComplexSPE}. Besides the envelop soliton solution, the CSP equation possesses rogue wave solution of any order in analogue to the nonlinear Schr\"odinger (NLS) equation~\cite{FLZPhysD}.

The CSP equation can be viewed as an analogue of the NLS equation
in the ultra-short regime when the width of optical pulse is of the order $10^{-15} s$.  It is well known that
the NLS equation describes the evolution of slowly varying wave packets waves in weakly nonlinear dispersive media
under quasi-monochromatic assumption, which has been very successful in
many applications such as nonlinear optics and water waves~\cite{Yarivbook,Kodamabook,Agrawalbook,Ablowitzbook}. However, as the
width of optical pulse is in the order of femtosecond ($10^{-15}$ s),
then the width of spectrum of this ultra-short pulses is approximately
of the order $10^{15} s^{-1}$, the monochromatic assumption to derive
the NLS equation is not valid~\cite{Rothenberg}. Description of
ultra-short pulses
requires a modification of standard slow varying envelope models based
on the NLS equation.
That is the motivation for the study of the short pulse equation, the CSP equation and their coupled models.

The CSP equation is mathematically related to a two-component short
pulse (2-SP) equation proposed by Dimakis and
M\"uller-Hoissen~\cite{Hoissen_CSPE}
and Matsuno~\cite{Matsuno_CSPE} independently. If we
take $u={\text{Re}} (q)$ and $v={\text{Im}} (q)$, then the 2-SP equation
in \cite{Hoissen_CSPE,Matsuno_CSPE} becomes the CSP equation. We showed
that the CSP equation can be derived from the motion of space curves and provide an
alternative multi-soliton solution in terms of determinant based on the KP
hierarchy reduction technique~\cite{FengShen_ComplexSPE}. The multi-breather and higher order rogue wave solutions to the CSP equation were constructed by the Darboux transformation method in \cite{FLZPhysD}. Furthermore, we have constructed integrable semi-discrete CSP equation and apply it as a self-adaptive moving mesh method for the numerical simulations of the CSP equation~\cite{FMO-PJMI}.

Since the NLS equation has the focusing and defocusing cases, which admits the bright and dark type soliton solutions, respectively. It is natural that the CSP equation can also have the focusing and defocusing type, which may be proposed as
\begin{equation}
q_{xt}+q+\frac{1}{2} \sigma \left( |q|^{2}q_{x}\right) _{x}=0\,,  \label{gCSP}
\end{equation}
where $\sigma=1$ represents the focusing case, and $\sigma=-1$ stands for the defocusing case.
It turns out that this is indeed the case.
Same as the focusing CSP equation discussed in \cite{Feng_ComplexSPE,FengShen_ComplexSPE}, the defocusing CSP equation can also
occur in nonlinear optics when ultra-short pulses propagate in a nonlinear
media of defocusing type~\cite{FenglingzhuPRE}.

In the present paper, we will study the defocusing CSP
equation
\begin{equation}
q_{xt}+q-\frac{1}{2}\left( |q|^{2}q_{x}\right) _{x}=0\,,
\label{dCSP}
\end{equation}
both geometrically and algebraically. The goal of the present paper is twofold.
The one is to investigate the geometric meaning of the defocusing CSP
equation,
especially, its connection with the motion of space curves.
The other is to find out how these
equations are reduced from the
extended KP hierarchy and to construct their
$N$-soliton solutions. The remainder of this paper is organized as follows. In
section 2, we firstly establish the connection between the motions of space curves in Minkowski space $\mathbf{R}^{2,1}$
and the defocusing CSP equation via a hodograph (reciprocal) transformation.
The Lax pair is constructed geometrically to assure the integrability of
the defocusing CSP equation. Then, starting from the fundamental forms
of the surfaces embedded in $\mathbf{R}^3$ and $\mathbf{R}^{2,1}$, the focusing and
defocusing complex coupled dispersionless (CCD) system
are derived, respectively. The curve flows are also made clear.
In section 3, starting from a set
of bilinear equations for the single-component extended KP hierarchy, as well as
their tau functions, we deduce the defocusing CSP equation by the KP hierarchy reduction method. Meanwhile, as a
by-product, the $N$-dark soliton solution is obtained. Section 4 is
devoted to concluding remarks.

\section{Geometric Formulations}
It has been known for several decades that there are deep connections between the differential geometry and the theory of integrable systems, and various integrable differential or difference equations arise from either curve dynamics or surfaces. For example, the sine-Gordon (sG) equation, the modified KdV (mKdV) equation and the NLS equation from the compatibility conditions of the motion of either plane or space curves, as shown by many researchers including Lamb and Hasimoto~\cite{Lamb,Hasimoto,GoldsteinPRL,NakayamaPRL,DoliwaPLA94,NakayamaJPSJ98,Calini00,Ivey,QuChouPhyD}. On the other hand, a more broad class of soliton equations including the ones mentioned above can also be derived from the theory of surfaces (see the pioneer work in \cite{Sasaki79,Chern80,Symsoliton,Tafel,SymVII,Symloopsoliton,Terng97} and also the book by Roger and Schief~\cite{RogerSchiefbook}). In this section, a link between the CSP equation and the motion of curves, as well as surfaces, in three-dimensional space is established.
\subsection{The link with the motion of space curves in Minkowski space}
In the study of curve flows of soliton equations, it has shown in the past that, very often, we have to turn from the Euclidean space to Minkowski space when we attempt to find the links between some soliton equations such as the defocusing NLS equation~\cite{NakayamaJPSJ98,Anco16}. This is also the case when we attempt to establish a link of the defocusing CSP equation (\ref{dCSP}) to the motion of space curves, as will be shown in this subsection.
Firstly, noting 
that the CSP equation (\ref{gCSP}) admits the following conservative law
\begin{equation} \label{gCSP_conv}
(\sqrt{1+\sigma|q_x|^2})_t + \frac{1}{2} \sigma (|q|^2\sqrt{1+\sigma|q_x|^2})_x=0\,,
\end{equation}
which allows us to define a hodograph (reciprocal) transformation
\begin{equation}
\mathrm{d}s=\rho^{-1} \mathrm{d}x - \frac{1}{2} \sigma  \rho^{-1} {|q|^2} \mathrm{d%
}t,\quad \, \mathrm{d}y=-\mathrm{d}t,
\end{equation}
where $\rho^{-1}=\sqrt{1+\sigma|q_x|^2}$.
By doing so, one can obtain the differential conversion formula between $(x,t)$ and $(y,s)$
\begin{equation}
\partial_x =  \rho^{-1} \partial_s\, \quad \partial_t=-\partial_y- \frac{1}{2}  \sigma \rho^{-1} {|q|^2} \partial_s\,,
\end{equation}
or
\begin{equation}
\partial_s =  \rho \partial_x\, \quad \partial_y=-\partial_t- \frac{1}{2}  \sigma {|q|^2} \partial_x\,.
\end{equation}
Therefore, the CSP equation  (\ref{gCSP}) is converted into
\begin{equation}  \label{CCDa}
q_{ys}=\rho q\,,
\end{equation}
while the conservative form of the CSP (\ref{gCSP_conv})
\begin{equation}  \label{CCDb}
\rho_y + \frac{1}{2}\sigma (|q|^2)_s=0\,.
\end{equation}

We remark here that equations (\ref{CCDa}) and (\ref{CCDb}) constitute a
coupled nonlinear system. When $s$ is viewed as a spatial variable, and
$y$ as a temporal variable, the system for $\sigma=1$ is called the
focusing complex coupled dispersionless (CCD) system, which has been
studied in \cite{KonnoKakuhata2} and related references.
For some reason, the system (\ref{CCDa}) and (\ref{CCDb}) for
$\sigma=-1$ has been overlooked in the past and its soliton solution has not
been studied yet. In what follows, we reformulate the system
(\ref{CCDa}) and (\ref{CCDb}) with $\sigma=-1$ geometrically
in order to make clear the geometric interpretation of the defocusing CSP equation (\ref{dCSP}).

It can be easily checked that the quantity $\rho^2+\sigma |q_s|^2$ is independent of $y$, thus, we can assume
$\rho^2+\sigma |q_s|^2=1$ without loss of generality. Under this case, if we assume $Q=q_s$, the system can be simplified into a single equation
\begin{equation}  \label{CsG_general}
\left(\frac{Q_y}{\sqrt{1-\sigma|Q|^2}} \right)_s=Q\,,
\end{equation}
which is called the complex sine-Gordon equation for $\sigma=1$~\cite{Symsoliton,SymVII}. Eq. (\ref{CsG_general}) is also a reduction of a so-called vector sine-Gordon equation~\cite{JWang1}. If $Q=q_s$ is a real-valued function,
the complex sine-Gordon equation (\ref{CsG_general}) with $\sigma=1$
leads to the sine-Gordon equation $\theta_{ys}=\sin \theta$ by setting $Q=q_s=\sin
\theta$ and $\rho=\cos \theta$.
In the case of $\sigma=-1$, if $Q=q_s$ is a real-valued function, we
obtain
the sinh-Gordon equation $\theta_{ys}=\sinh \theta$ from (\ref{CsG_general}) by setting $Q=q_s=\sinh
\theta$ and $\rho=\cosh \theta$. Thus we can call (\ref{CsG_general}) the
complex sine-Gordon equation for $\sigma=1$ and complex sinh-Gordon equation for $\sigma=-1$.

It was pointed out in \cite{FengShen_ComplexSPE,FMO-PJMI} that the
focusing CCD system
is linked to the focusing CSP equation by a
hodograph (reciprocal) transformation, by which the Lax pair of the
focusing CSP equation was established.
Both the generalized CD equation and the focusing CSP equation were
interpreted
as the motion of space curve in Euclidean
space~\cite{FengShen_ComplexSPE}.
In what follows, we proceed to the study for the relationship
between the defocusing CSP equation (\ref{dCSP}) and the motion of space curves.
To this end,  we need to turn to consider the motion of curves which
lies in a surface $S$ embedded in a Minkowski space $\mathbf{R}^{2,1}$ equipped
with a Lorentz metric
$dl^2=-dx_{1}^{2}+dx_{2}^{2}+dx_{3}^{2}$.
For any $\vec{x}=(x_{1},x_{2},x_{3})$, $\vec{y}=(y_{1},y_{2},y_{3})$ in $\mathbf{R}^{2,1}$, the scalar product
is defined as $\langle \vec{x},\vec{y}\rangle=-x_{1}y_{1}+x_{2}y_{2}+x_{3}y_{3}$ and the vector
product is defined as $\vec{x}\times \vec{y}=(x_{3}y_{2}-x_{2}y_{3},x_{3}y_{1}-x_{1}y_{3},x_{1}y_{2}-x_{2}y_{1})$.
Moreover, we use the Darboux frame $\{\mathbf{T},\mathbf{N},\mathbf{t}\}$ attached to a curve $\vec{\textbf{r}} (y,s)$ parameterized by the arc-length:
\begin{eqnarray*}
\mathbf{T} &=&\vec{\textbf{r}}_{y}\qquad \text{ (the unit tangent vector)}\,, \\
\mathbf{N} &=&\mathbf{N(\vec{\textbf{r}} )}\qquad \text{ (the unit normal vector
 of a surface $S$)}\,, \\
\mathbf{t} &=&\mathbf{N\times T}\qquad \text{(the tangent normal vector)}\,,
\end{eqnarray*}%
where $y$ stands for the arc length and $s$ represents the time.
The general equation (the Darboux equation) for the orthogonal triad $\{\mathbf{T},\mathbf{N},\mathbf{t}\}$ along the curve takes
the form
\begin{equation}
\left[
\begin{array}{c}
\mathbf{T} \\
\mathbf{t} \\
\mathbf{N}%
\end{array}%
\right] _{y}=\left[
\begin{array}{ccc}
0 & \kappa _{g} & \kappa _{n} \\
\kappa _{g} & 0 & \tau _{r} \\
\kappa _{n} & -\tau _{r} & 0%
\end{array}%
\right] \left[
\begin{array}{c}
\mathbf{T} \\
\mathbf{t} \\
\mathbf{N}%
\end{array}%
\right]\,,
\end{equation}%
where $\kappa_g$ is the geodesic curvature, $\kappa_n$ is the normal
curvature and $\tau_r$ is the relative torsion (geodesic torsion)
while the general temporal evolution of $\gamma$ can be expressed as
\begin{equation}
\left[
\begin{array}{c}
\mathbf{T} \\
\mathbf{t} \\
\mathbf{N}%
\end{array}%
\right] _{s}=\left[
\begin{array}{ccc}
0 & \alpha  & \beta  \\
\alpha  & 0 & \gamma  \\
\beta  & -\gamma  & 0%
\end{array}%
\right] \left[
\begin{array}{c}
\mathbf{T} \\
\mathbf{t} \\
\mathbf{N}%
\end{array}%
\right]\,.
\end{equation}%
The compatibility conditions lead to the following system
\begin{equation}
\kappa_{g,s}=\alpha _{y}+\kappa _{n}\gamma -\tau _{r}\beta \,,  \label{3Dcurve_Integrabilitym1}
\end{equation}%
\begin{equation}
\kappa_{n,s}=\beta _{y}-\kappa _{g}\gamma +\tau _{r}\alpha \,,
\label{3Dcurve_Integrabilitym2}
\end{equation}%
\begin{equation}
\tau_{r,s}=\gamma _{y}-\kappa _{g}\beta +\kappa _{n}\alpha \,.  \label{3Dcurve_Integrabilitym3}
\end{equation}%
Combining Eq.(\ref{3Dcurve_Integrabilitym1}) with Eq.(\ref{3Dcurve_Integrabilitym2}), we have
\begin{equation}
(\kappa _{g}+\mathrm{i}\kappa _{n})_{s}=(\alpha +\mathrm{i}\beta )_{y}-%
\mathrm{i}\gamma (\kappa _{g}+\mathrm{i}\kappa _{n})+\mathrm{i}\tau
_{r}(\alpha +\mathrm{i}\beta )\,. \label{3Dcurve_Integrabilitym4}
\end{equation}
If we choose
\begin{equation}
\kappa _{g}+\mathrm{i}\kappa _{n}=-\mathrm{i}q\,,\quad \tau _{r}=-c^{-1} \,,
\label{CCD_assup1}
\end{equation}%
\begin{equation}
\alpha +\mathrm{i}\beta =cq_{s}\,,\quad \gamma =c\rho\,,
\label{CCD_assup2}
\end{equation}
then Eq.(\ref{3Dcurve_Integrabilitym4}) becomes
\begin{equation}
\label{CCDa1}
q_{ys}=\rho q.
\end{equation}%
On the other hand, since
\begin{eqnarray*}
\frac{1}{2}\left( |q|^{2}\right) _{s} &=&\frac{1}{2}(qq_{s}^{\ast
}+q_{s}q^{\ast }), \\
&=&\frac{1}{2c}\left( \mathrm{i}(\alpha -\mathrm{i}\beta )(\kappa _{g}+%
\mathrm{i}\kappa _{n})-\mathrm{i}(\kappa _{g}-\mathrm{i}\kappa _{n})(\alpha +%
\mathrm{i}\beta )\right) , \\
&=&-\frac{1}{c}(\kappa _{n}\alpha -\kappa _{g}\beta ),
\end{eqnarray*}%
one has
\begin{equation}
\rho_{y}-\frac{1}{2}(|q|^{2})_{s}=0\,,
\label{CCDb1}
\end{equation}%
from Eq.(\ref{3Dcurve_Integrabilitym3}). Thus the link of the defocusing CCD
system to the motion of space curves in Minkowski space $\mathbf{R}^{2,1}$ is established.
Recall the Lie group
\begin{equation*}
SU(1,1)= \{g \in SL(2, C) \mid g^*Jg=J\}\,,
\end{equation*}%
where $J={\text{diag}}(1, -1)$, and the Lie algebra
\begin{equation*}
su(1,1)= \left\{
\begin{pmatrix}
\mathrm{i}x  & z \\
z^* & -\mathrm{i}x
\end{pmatrix} \mid x  \in R, z \in C \right\}\,.
\end{equation*}%
We choose the basis of $su(1,1)$ as
\begin{equation*}
\mathrm{\tilde{e}}_{1}=\frac{1}{2}\left(
\begin{array}{cc}
\mathrm{i} & 0 \\
0 & -\mathrm{i}%
\end{array}%
\right) ,\quad \mathrm{\tilde{e}}_{2}=\frac{1}{2}\left(
\begin{array}{cc}
0 & 1 \\
1 & 0%
\end{array}%
\right) ,\quad \mathrm{\tilde{e}}_{3}=\frac{1}{2}\left(
\begin{array}{cc}
0 & \mathrm{i} \\
-\mathrm{i} & 0%
\end{array}%
\right)\,,
\end{equation*}
which satisfies the communication relation
\begin{equation*}
[\mathrm{\tilde{e}}_{1}, \mathrm{\tilde{e}}_{2}]=\mathrm{\tilde{e}}_{3}\,, \quad
[\mathrm{\tilde{e}}_{2},\mathrm{\tilde{e}}_{3}]=-\mathrm{\tilde{e}}_{1}\,, \quad
[\mathrm{\tilde{e}}_{3}, \mathrm{\tilde{e}}_{1}]=\mathrm{\tilde{e}}_{2}\,.
\end{equation*}
We identify $\mathbf{R}^{2,1}$ as $su(1,1)$ by $r=-Z \mathrm{\tilde{e}}_{1} +Y \mathrm{\tilde{e}}_{2}-X \mathrm{\tilde{e}}_{3} \rightarrow \vec{\mathbf{r}}=(X,Y,Z)^T$, then
\begin{equation}
\langle \vec{r}, \vec{w}\rangle = -2{\text{tr}}(rw)\,, \quad \vec{r} \times \vec{w}= r \times w\,.
\end{equation}
On the other hand, recall the Lie group
\begin{equation*}
SO(1,2)= \{A \in SL(3) \mid A^TI_{1,2}A=I_{1,2}\}\,,
\end{equation*}%
where $I_{1,2}={\text{diag}} (-1, 1,1)$, and the Lie algebra
\begin{equation*}
so(1,2)=
\left\{ \begin{pmatrix}
0 & X & Y \\ X & 0 & Z \\ Y & -Z & 0
\end{pmatrix} \mid X, Y,Z \in R \right\}\,,
\end{equation*}%
we choose the basis of $so(1,2)$ as
\begin{equation*}
\mathrm{\tilde{L}}_{1}=\left(
\begin{array}{ccc}
0 & 0 & 0 \\
0 & 0 & -1 \\
0 & 1 & 0%
\end{array}%
\right) ,\quad
\mathrm{\tilde{L}}_{2}=\left(
\begin{array}{ccc}
0 & 0 & 1 \\
0 & 0 & 0 \\
1 & 0 & 0%
\end{array}%
\right) ,\quad
\mathrm{\tilde{L}}_{3}=\left(
\begin{array}{ccc}
0 & -1 & 0 \\
-1 & 0 & 0 \\
0 & 0 & 0%
\end{array}%
\right)\,,
\end{equation*}
which satisfies the communication relation
\begin{equation*}
[\mathrm{\tilde{L}}_{1}, \mathrm{\tilde{L}}_{2}]=\mathrm{\tilde{L}}_{3}\,, \quad
[\mathrm{\tilde{L}}_{2},\mathrm{\tilde{L}}_{3}]=-\mathrm{\tilde{L}}_{2}\,, \quad
[\mathrm{\tilde{L}}_{3}, \mathrm{\tilde{L}}_{1}]=\mathrm{\tilde{L}}_{2}\,.
\end{equation*}
We identify $\mathbf{R}^{2,1}$ as $so(1,2)$ by $r=-Z \mathrm{\tilde{L}}_{1} +Y \mathrm{\tilde{L}}_{2}-X \mathrm{\tilde{L}}_{3}
\rightarrow \vec{\mathbf{r}}=(X,Y,Z)^T$.
Obviously, there is an isomorphism between the Lie algebras $su(1,1)$ and
$so(1,2)$ which is reflected by the correspondence $\mathrm{\hat{L}}_{j}\leftrightarrow \mathrm{\tilde{e}}_{j}$ ($%
j=1,2,3$).
Based on this fact, we can easily construct the Lax pair for the defocusing CCD system
geometrically as follows
\begin{equation}
\Psi _{y}=U\Psi ,\quad \Psi _{s}=V\Psi \,,
\end{equation}%
where
\begin{eqnarray}
U &=&-\kappa _{g}\mathrm{\tilde{e}}_{3}+\kappa _{n}\mathrm{\tilde{e}}_{2}-\tau _{r}
\mathrm{\tilde{e}}_{1}  \nonumber\\
&=&\frac{1}{2}\left(
\begin{array}{cc}
-\mathrm{i}\tau _{r} & \kappa _{n}-\mathrm{i}\kappa _{g} \\
\kappa _{n}+\mathrm{i}\kappa _{g} & \mathrm{i}\tau _{r}%
\end{array}%
\right) \, \nonumber\\
&=&  \left(
\begin{array}{cc}
\frac 12 \mathrm{i}\lambda & -\frac{1}{2}q \\
-\frac{1}{2}q^{\ast } & -\frac 12 \mathrm{i}\lambda%
\end{array}%
\right)\,,
\end{eqnarray}%
\begin{eqnarray}
V &=&-\alpha \mathrm{\tilde{e}}_{3}+\beta \mathrm{\tilde{e}}_{2}-\gamma\mathrm{\tilde{e}}_{1} \nonumber\\
&=&\frac{1}{2}\left(
\begin{array}{cc}
-\mathrm{i}\gamma & \beta -\mathrm{i}\alpha \\
\beta +\mathrm{i}\alpha & \mathrm{i}\gamma%
\end{array}%
\right) \, \nonumber\\
&=& -\frac{\mathrm{i}}{2\lambda}\left(
\begin{array}{cc}
\rho & q_{s} \\
-q_{s}^{\ast } & -\rho%
\end{array}%
\right)\,,
\end{eqnarray}
by setting $c^{-1}=\lambda$.
The above Lax pair is consistent with the
one used for constructing the Darboux transformation of the defocusing CCD system~\cite{FenglingzhuPRE}. The Lax pair found geometrically here shows that the CCD system is simply the negative order of the AKNS hierarchy~\cite{CsG1,PavlovPLA,DajunPhysD}.

It is known that there exists a relationship between the Frenet-Serret
frame and the Darboux frame in 3-dimensional Euclidean space $\mathbf{R}^3$:
\begin{equation*}
\kappa_g = \kappa \cos \alpha\,, \ \ \kappa_n = -\kappa \sin \alpha\,, \ \ \frac{d \alpha}{d y} = \tau -\tau_r\,,
\end{equation*}
where $\alpha$ is the rotation angle in the tangent plane from the Frenet-Serret frame to the Darboux frame. Therefore,
we have
\begin{equation}
q= \mathrm{i} (\kappa_g + \mathrm{i} \kappa_n) = \kappa e^{-\mathrm{i} \alpha + \mathrm{i} \pi/2} = \kappa e^{-\mathrm{i} \int \tau dy + \mathrm{i} (c^{-1}y+\pi/2)} \,.
\end{equation}
The above formula can be viewed as Hasimoto transformation for the case of the CCD system.

The reciprocal link between the CCD system and the
CSP equation can be defined geometrically.
Putting $c=1$, we define
\begin{equation}
x=Z=\int_{s_0}^{s}\gamma (y,s^{\prime })ds^{\prime
}=\int_{s_0}^{s}\rho (y,s^{\prime })ds^{\prime }\,,\quad t=-y\,,
\end{equation}
and
\begin{equation}
q=X+\mathrm{i}Y=\int_{s_0}^{s} (\alpha +\mathrm{i}\beta) ds^{\prime
}=\int_{s_0}^{s} q_{s^{\prime}}(y,s^{\prime})ds^{\prime }\,,
\end{equation}
where $s_0$ in the arc-length parameter $s$ corresponds to the origin of $x$-coordinate.
It can be easily shown that
\begin{equation*}
\frac{\partial x}{\partial s}=\rho \,, \ \ \frac{\partial x}{\partial y}=\int_{s_0}^{s}\rho _{y}(y,s^{\prime })ds^{\prime }=\frac{1}{2}\int_{s_0}^{s}(|q|^{2})_{s^{\prime}}ds^{\prime }=\frac{1}{2}|q|^{2}\,,
\end{equation*}%
which realize the hodograph (reciprocal) transformation mentioned in the
previous section. Therefore, we have a geometric interpretation for the CSP equation, that is, the CSP equation represents the same integrable curve flow as the CCD system in either $\mathbf{R}^3$ (focusing case) or $\mathbf{R}^{2,1}$ (defocusing case), when $Z$-coordinate becomes one independent variable (spatial one) via the hodograph (reciprocal) transformation, $X$- and $Y$- coordinates are interpreted as the real part and imaginary part of the dependent variable $q$.

Under this hodograph (reciprocal) transformation, we can obtain the Lax pair for the defocusing CSP equation (\ref{dCSP}) as follows
\begin{equation}
\Psi _{x}=P\Psi ,\quad \Psi _{t}=Q\Psi \,,
\end{equation}%
where
\begin{equation}
P=\rho ^{-1}V=-\frac{\mathrm{i}}{2\lambda} \left(
\begin{array}{cc}
1 & q_{x} \\
-q_{x}^{\ast } & -1%
\end{array}%
\right) \,,
\end{equation}

\begin{equation}
Q=\frac{1}{2}|q|^{2}P-U=\left(
\begin{array}{cc}
-\frac{\mathrm{i}}{2} \lambda-\frac{\mathrm{i}}{4\lambda}|q|^{2} &
-\frac{\mathrm{i}}{4\lambda}|q|^{2}q_{x}+\frac 12 q \\
\frac{\mathrm{i}}{4\lambda}|q|^{2}q_{x}^{\ast}+\frac{1}{2}q^{\ast} &
\frac{\mathrm{i}}{2} \lambda+\frac{\mathrm{i}}{4\lambda }|q|^{2}%
\end{array}%
\right) \,.
\end{equation}
\subsection{The link with surfaces embedded in space}
Prior to the further pursuit of geometric meaning of the defocusing CSP
 equation, we turn to reveal the geometric interpretation of the
 focusing CCD system. If we interchange $y$ and $s$ and take $\lambda
 \to - 1/(2\lambda)$, then the Lax representation for the CCD system
 of focusing type~\cite{FengShen_ComplexSPE} can be cast into
\begin{equation*}
\Psi _{y}=U\Psi ,\quad \Psi _{s}=V\Psi \,,
\end{equation*}%
\begin{equation*}
U =\left(
\begin{array}{cc}
-\frac{1}{2} \mathrm{i} \lambda & -\frac{1}{2}q \\
\frac{1}{2}q^{\ast } & \frac{1}{2} \mathrm{i} \lambda %
\end{array}%
\right) \,, \quad
V =\frac{\mathrm{i}}{2\lambda} \left(
\begin{array}{cc}
\rho & q_s \\
 q_s^{\ast}& -\rho%
\end{array}%
\right)
\end{equation*}%
Since $\rho^2 + |q_s|^2=1$ so we can assume
\begin{equation}
\rho =\cos \theta\,, \quad q_s= \sin \theta e^{-\mathrm{i} \omega}\,,
\end{equation}
and
\begin{equation}
q=(\theta _{y}-\mathrm{i}\omega _{y}\tan \theta) e^{-\mathrm{i}\omega}\,.
\end{equation}
Then the matrices $U$ and $V$  become
\begin{eqnarray}
U &=& \left(
\begin{array}{cc}
-\frac{1}{2} \mathrm{i} \lambda & -\frac 12(\theta _{y}-\mathrm{i}\omega _{y}\tan \theta) e^{-\mathrm{i}\omega}  \\
\frac 12(\theta _{y}+\mathrm{i}\omega _{y}\tan \theta) e^{\mathrm{i}\omega} & \frac{1}{2} \mathrm{i} \lambda
\end{array}%
\right) ,
\end{eqnarray}

\begin{eqnarray}
V&=&\frac{\mathrm{i}}{2\lambda} \left(
\begin{array}{cc}
\cos \theta & \sin \theta e^{-\mathrm{i}\omega } \\
\sin \theta e^{\mathrm{i}\omega } & -\cos \theta%
\end{array}%
\right)\,.
\end{eqnarray}
It has been known, for any integrable system possessing a $su(2)$ linear
representation,
we can come up with fundamental forms of surfaces~\cite{Symsoliton,RogerSchiefbook}, which read
\begin{equation}
  \text{I}= dy^2 + 2 \cos \theta dy ds + ds^2\,,
\label{Form1}
\end{equation}
\begin{equation}
 \text{II} = (\tan \theta) \omega_y dy^2 + 2 \sin \theta dy ds +(\sin \theta) \omega_s ds^2  \,.
\label{Form2}
\end{equation}
The resulting first fundamental form represents a Chebyschev net for a curve $\textbf{r}(y,s)$ embedded on the surface $\Sigma$. $\theta$ represents the angle between $\textbf{r}_y$ and $\textbf{r}_s$.
The associated Gauss equations read
 \begin{eqnarray}
 \textbf{r}_{yy} &=& (\cot \theta) \theta_y \textbf{r}_{y} -(\csc \theta) \theta_y \textbf{r}_{s} -(\tan \theta) \omega_y \textbf{N} \,,  \label{Gauss_eq1} \\
     \textbf{r}_{ys} &=& \sin \theta \textbf{N}  \,,   \label{Gauss_eq2}\\
  \textbf{r}_{ss} &=& -(\csc \theta) \theta_s \textbf{r}_{y}+(\cot \theta) \theta_s \textbf{r}_{s} +(\sin \theta) \omega_s \textbf{N}  \label{Gauss_eq3} \,,
 \end{eqnarray}
while the Weingarten equations are
 \begin{eqnarray}
 \textbf{N}_{y} &=& (\cot \theta + \csc \theta \sec \theta \omega_y) \textbf{r}_{y}-(\csc \theta \omega_y +\csc \theta ) \textbf{r}_{s}  \,, \label{Weigarten_eq1} \\
   \textbf{N}_{s} &=&  -(\csc \theta - \cot \theta \omega_s) \textbf{r}_{y}+(\cot \theta + \csc \theta \omega_s ) \textbf{r}_{s} \,. \label{Weigarten_eq2}
 \end{eqnarray}
The Mainardi-Codazzi equations give
\begin{equation}
(\omega_s \cos \theta)_y=\left( \frac{\omega_y}{\cos \theta}\right)_s\,.
\label{complexsG2}
\end{equation}
The Gaussian curvature is
\begin{equation}
K=-\frac{(\tan \theta) \omega_{y}\omega _{s}+\sin\theta}{\sin \theta}\,,
\label{Gaussian_curvature}
\end{equation}
so the Liouville-Beltrami form  of the  \textit{Theorema egregium} takes
\begin{equation}
\theta _{ys}-\sin \theta -(\tan \theta) \omega_{y}\omega _{s}=0\,.
\label{complexsG1}
\end{equation}
The system (\ref{complexsG2}) and (\ref{complexsG1}) is an alternative form of the focusing CCD system.
As mentioned in \cite{RogerSchiefbook},  it is also equivalent to the self-induced transparency (SIT) equations~\cite{SteudalPLA} and an integrable model for the stimulated Raman scattering (SRS)~\cite{KaupPhysD,SteudalPhysD}.

Let us assume the position vector on the surface
$$\textbf{r}(y,s)=(X(y,s),Y(y,s),Z(y,s)),$$
which can be represented as a matrix form
\begin{equation} \label{positionvec}
{r}(y,s)=X(y,s) \mathrm{e}_1 + Y(y,s)\mathrm{e}_2+Z(y,s) \mathrm{e}_3 \,
\end{equation}
by referring to the isomorphism between $su(2)$ and $SO(3)$ in $\mathbf{R}^3$. Moreover, we define
\begin{equation}
 \textbf{T} = \Phi^{-1} \mathrm{e}_3 \Phi \,, \quad  \textbf{N} = \Phi^{-1}\mathrm{e}_2 \Phi \,, \quad  \textbf{t} = \Phi^{-1} \mathrm{e}_1 \Phi \,,
\end{equation}
where $\mathrm{e}_1, \mathrm{e}_2, \mathrm{e}_3$ are expressed as
$\mathrm{e}_i=\frac{1}{2\mathrm{i}}\sigma_i$ for $i=1,2,3$ by using the Pauli matrices.
Since
\begin{equation}
\textbf{r}_y= \left. \Phi^{-1} U_\lambda \Phi\right|_{\lambda=1}\,,
\end{equation}
we then have
\begin{equation} \label{r_y}
\textbf{r}_y= \Phi^{-1} \mathrm{e}_3 \Phi = \textbf{T}\,,
\end{equation}
which coincides with the assumption of Darboux frame and $y$ plays a role of arc length.
From
\begin{equation}
\textbf{r}_s= \left. \Phi^{-1} V_\lambda \Phi\right|_{\lambda=1}\,,
\end{equation}
we have
\begin{equation} \label{r_u}
\textbf{r}_s= (\cos \theta) \textbf{T} + (\sin \theta \cos \omega) \textbf{N} + (\sin \theta \sin \omega) \textbf{t}\,,
\end{equation}
which can also written as
\begin{equation} \label{r_curve}
\textbf{r}_s= (\cos \theta) \textbf{r}_y +  \textbf{r}_y \times \textbf{r}_{ys}\,.
\end{equation}
This gives the curve flow for the focusing CCD system. Furthermore, based on the Darboux transform~\cite{FLZPhysD} of
\begin{eqnarray}
U &=&  \left(
\begin{array}{cc}
e^{-\frac 12 \mathrm{i}\lambda y + \frac{\mathrm{i}}{2\lambda} s} & 0 \\
0 & e^{\frac 12 \mathrm{i}\lambda y - \frac{\mathrm{i}}{2\lambda} s}
\end{array}%
\right) \,,
\end{eqnarray}
and the Sym-Tafel formula~\cite{Symsoliton,Tafel}
\begin{equation} \label{SymTafel}
r=\left.\Phi^{-1}\Phi_{\lambda}\right|_{\lambda=1}\,,
\end{equation}
we can calculate the one-soliton surface as follows
\begin{eqnarray}
  X &=& \frac{b}{(1-a)^2+b^2} \text{sech} R \cos W\,,  \\
  Y &=& \frac{b}{(1-a)^2+b^2} \text{sech} R \sin W\,,  \\
  Z &=& \frac{b}{(1-a)^2+b^2} \tanh R   +y + s\,,
\end{eqnarray}
where
\begin{equation*}
  R=b y + \frac{b}{a^2+b^2}s\,, \quad W=(1-a) y +\left(1+ \frac{a}{a^2+b^2} \right)s
\end{equation*}
The surface for a moving one-soliton with parameter $a=0.4, b=1$ is shown in Fig. 1.
\begin{figure}[htb]
\centering
\includegraphics[height=60mm,width=80mm]{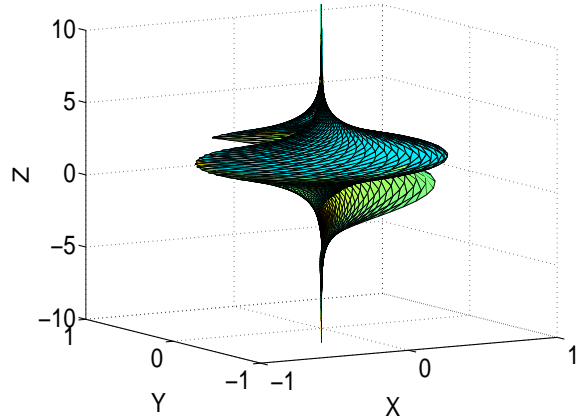}
\caption{One-soliton surface for $a=0.4, b=1$.}
\label{fig1}
\end{figure}
As mentioned in \cite{RogerSchiefbook}, the system (\ref{complexsG2}) and (\ref{complexsG1}) is directly related to
the so-called Pohlmeyer--Lund--Regge system which was originally demonstrated by Lund and Regge~\cite{LundRegge76} to represent the relativistic motion
of a string in a uniform and static external field, and by Pohlmeyer~\cite{PohlmeyerCMP76} to represent a $O(4)$ nonlinear sigma model.
It was shown by Lund~\cite{Lund77,Lund78} that the
Pohlmeyer--Lund--Regge system can also be interpreted as the
Gauss--Mainardi--Codazzi equations for particular surfaces in $S^3$ and
can be solved by the inverse scattering transformation (IST) method.
From (\ref{r_curve}), it is obvious that
\begin{equation} \label{PLG_curve}
\textbf{r}_{ys}= \textbf{r}_y \times \textbf{r}_{s}\,.
\end{equation}


We proceed to establish a link of the defocusing CCD system with the surfaces. Similarly, we could obtain
fundamental forms of surfaces~\cite{Symsoliton,RogerSchiefbook} embedded in Minkowski space $\mathbf{R}^{2,1}$
\begin{equation}
  \text{I}= dy^2 + 2 \cosh \theta dy ds + ds^2\,,
\label{Form11}
\end{equation}
\begin{equation}
 \text{II} = (\tanh \theta) \omega_y dy^2 + 2 \sinh \theta dy ds +(\sinh \theta) \omega_s ds^2  \,.
\label{Form22}
\end{equation}
Here $\theta$ represents the angle between $\textbf{r}_y$ and $\textbf{r}_s$ in Minkowski space~\cite{Lopez}.
Therefore, the associated Gauss equations take the form
 \begin{eqnarray}
 \textbf{r}_{yy} &=& (\coth \theta) \theta_y \textbf{r}_{y} -(\sinh \theta)^{-1} \theta_y \textbf{r}_{s} -(\tanh \theta) \omega_y \textbf{N} \,,  \label{Gauss_eq11} \\
     \textbf{r}_{ys} &=& \sinh \theta \textbf{N}  \,,   \label{Gauss_eq22}\\
  \textbf{r}_{ss} &=& -(\sinh \theta)^{-1} \theta_s \textbf{r}_{y}+(\coth \theta) \theta_s \textbf{r}_{s} +(\sinh \theta) \omega_s \textbf{N}  \label{Gauss_eq33} \,,
 \end{eqnarray}
while the Weingarten equations read
 \begin{eqnarray}
 \textbf{N}_{y} &=& (\coth \theta + (\sinh \theta \cosh \theta)^{-1} \omega_y) \textbf{r}_{y}-((\sinh \theta)^{-1} \omega_y +(\sinh \theta)^{-1} ) \textbf{r}_{s}  \,, \label{Weigarten_eq11} \\
   \textbf{N}_{s} &=&  -((\sinh \theta)^{-1} - \coth \theta \omega_s) \textbf{r}_{y}+(\coth \theta + (\sinh \theta)^{-1} \omega_s ) \textbf{r}_{s} \,. \label{Weigarten_eq22}
 \end{eqnarray}
Both of the Mainardi-Codazzi equations lead to the same equation
\begin{equation}
(\omega_s \cosh \theta)_y=\left( \frac{\omega_y}{\cosh \theta}\right)_s\,.
\label{complexshG2}
\end{equation}
Since the Gaussian curvature is
\begin{equation}
K=-\frac{(\tanh \theta) \omega_{y}\omega _{s}+\sinh\theta}{\sinh \theta}\,,
\label{Gaussian_curvature2}
\end{equation}
then the Liouville-Beltrami form  of the  \textit{Theorema egregium} becomes
\begin{equation}
\theta _{ys}-\sinh \theta -(\tanh \theta) \omega_{y}\omega _{s}=0\,.
\label{complexshG1}
\end{equation}
If we assume the following parameterizations
\begin{equation}
\rho = \cosh \theta\,, \quad q_s=\sinh \theta e^{-\mathrm{i} \omega}\,,
\end{equation}
and
\begin{equation}
q=(\theta _{y}-\mathrm{i}\omega _{y}\tanh \theta) e^{-\mathrm{i}\omega}\,,
\end{equation}
Then the system (\ref{complexshG2}) and (\ref{complexshG1}) becomes the defocusing CCD system.
On the other hand, if we assume
\begin{equation}
\omega _{s}=\chi _{s}\frac{\cosh \theta }{2\cosh ^{2}(\theta /2)}; \quad
\omega_{y}=\chi _{y}\frac{1}{2\cosh ^{2}(\theta /2)}\,,
\end{equation}
then the system (\ref{complexshG2}) and (\ref{complexshG1}) leads to the Pohlmeyer--Lund--Regge system of hyperbolic type~\cite{NakayamaJPSJ98}
\begin{equation}
\theta _{ys}-\sinh \theta -\frac{1}{2}\frac{\sinh \left( \theta /2\right) }{%
\cosh ^{3}\left( \theta /2\right) }\chi _{y}\chi _{s}=0\,,
\end{equation}

\begin{equation}
\chi _{ys}+\frac{1}{\sinh \theta }(\theta _{y}\chi _{s}+\theta _{s}\chi
_{y})=0\,.
\end{equation}
The Darboux frame can be cast into
\begin{equation}
 \textbf{T} = \Psi^{-1} \mathrm{\tilde{e}}_1 \Psi \,, \quad  \textbf{t} = \Psi^{-1}\mathrm{\tilde{e}}_2 \Psi \,, \quad  \textbf{N} = \Psi^{-1} \mathrm{\tilde{e}}_3 \Psi \,,
\end{equation}
then from
\begin{equation}
\textbf{r}_y= \left. \Phi^{-1} U_\lambda \Phi\right|_{\lambda=1}\,, \quad \textbf{r}_s= \left. \Phi^{-1} V_\lambda \Phi\right|_{\lambda=1}\,,
\end{equation}
we have
\begin{equation} \label{r_ys}
\textbf{r}_y= \textbf{T}\,, \ \
\textbf{r}_s= (\cosh \theta) \textbf{T} + (\sinh \theta \cos \omega) \textbf{N} + (\sinh \theta \sin \omega) \textbf{t}\,,
\end{equation}
the former coincides with the assumption of Darboux frame where $y$ serves as the arc length, the latter
gives the curve flow for the defocusing CCD system, which can also cast into
\begin{equation}
\textbf{r}_s= (\cosh \theta) \textbf{r}_y + \textbf{r}_y \times \textbf{r}_{ys} \,.
\end{equation}
\section{Reduction from the extended KP hierarchy}
\subsection{Bilinearizations of the defocusing CCD and CSP equations}
The bilinearizations of the defocusing CCD system and CSP equation are established
by the following propositions.
\begin{prop}
By means of the dependent variable transformations
\begin{equation} \label{var_tran1}
q= \frac{\beta}{2}\frac{g}{f} e^{\mathrm{i}(y+\gamma s/2)}\,,
\end{equation}
\begin{equation}  \label{var_tran2}
\rho= -\frac{\gamma}{2} -2(\log f)_{ys}\,,
\end{equation}
the defocusing CCD system (\ref{CCDa1})--(\ref{CCDb1}) are transformed into the following bilinear equations
for the tau functions $f$ and $g$
\begin{equation}
(D_{y}D_{s}+\mathrm{i} D_{s}+\frac{\gamma}{2} \mathrm{i} D_y)g\cdot f=0\,,  \label{CDSP_bilinear1}
\end{equation}
\begin{equation}
\left(D_{s}^{2}-\frac{\beta^2}{8}\right)f\cdot f=-\frac{\beta^2}{8}gg^{\ast}\,,  \label{CDSP_bilinear2}
\end{equation}
where $D$ is the Hirota $D$-operator defined by \cite{Hirotabook}
\begin{equation*}
D_s^n D_y^m f\cdot g=\left(\frac{\partial}{\partial s} -\frac{\partial}{%
\partial s^{\prime }}\right)^n \left(\frac{\partial}{\partial y} -\frac{%
\partial}{\partial y^{\prime }}\right)^m f(y,s)g(y^{\prime },s^{\prime
})|_{y=y^{\prime }, s=s^{\prime }}\,.
\end{equation*}
\end{prop}

\begin{proof}
The substitution of dependent variable transformation (\ref{var_tran2}) into Eq.(\ref{CCDb1}) yields
 \begin{equation*}
 -2(\log f)_{yss} - \frac{\beta^2}{8}  \left( \frac {gg^{\ast}}{f^2} \right)_y=0\,.
\end{equation*}
Integrating once in $y$ and setting the integration constant to be $\beta^2/8$, we then have
\begin{equation*}
 2(\log f)_{ss} - \frac{\beta^2}{8}  = -\frac{\beta^2}{8}  \frac {gg^{\ast}}{f^2}\,,
\end{equation*}
which is exactly the bilinear equation (\ref{CDSP_bilinear2}).
On the other hand, Eq.(\ref{CCDa1}) is converted into
\begin{equation}
  \left(\frac{g}{f}\right)_{y s} + \mathrm{i}  \left(\frac{g}{f}\right)_{s} + \mathrm{i}  \frac{\gamma}{2} \left(\frac{g}{f}\right)_{y} - \frac{\gamma}{2}  \frac{g}{f} = \left( -\frac{\gamma}{2}-2(\log f)_{ys}\right) \frac{g}{f} \,,
\end{equation}
via the dependent variable transformation (\ref{var_tran1}), or, is simplified into
\begin{equation}
  \left(\frac{g}{f}\right)_{y s} + 2 (\log f)_{ys} \frac{g}{f}   + \mathrm{i}  \left(\frac{g}{f}\right)_{s} + \mathrm{i}  \frac{\gamma}{2}  \left(\frac{g}{f}\right)_{y}  = 0 \,.
\end{equation}
Referring to a bilinear identity
\begin{equation*}
 \frac{D_y D_s g \cdot f}{f^2} =\left(\frac{g}{f}\right)_{y s} + 2 (\log f)_{y s} \frac{g}{f}  \,,
\end{equation*}
we then have the bilinear equation (\ref{CDSP_bilinear1}).
\end{proof}

\begin{prop}
By means of the dependent variable transformation
\begin{equation}
q=  \frac{\beta}{2} \frac{g}{f} e^{\mathrm{i}( y- s)}\,,
\end{equation}
and the hodograph (reciprocal) transformation
\begin{equation}
x =  -\frac{\gamma}{2} y + \frac{\beta^2}{8}  s -2(\log f)_s\,, \quad t=-s\,,
\end{equation}
the defocusing CSP equation (\ref{dCSP}) shares the same bilinear equations (\ref{CDSP_bilinear1})--(\ref{CDSP_bilinear2}).
\end{prop}

\begin{proof}
From the hodograph (reciprocal) transformation and bilinear equations, we could have
$$
\frac{\partial x}{\partial y} =  -\frac{\gamma}{2} -2(\log f)_{y s} = \rho\,,
$$
and
$$
\frac{\partial x}{\partial s} =  \frac{\beta^2}{8}  -2(\log f)_{ss} = \frac{\beta^2}{8}  \frac{|g|^2}{f^2} = \frac 12 |q|^2\,,
$$
which implies
\begin{equation*}
\partial _{y}=\rho \partial _{x},\quad \partial _{s}=-\partial_{t}+%
\frac{1}{2}|q|^{2}\partial _{x}\,.
\end{equation*}%
Thus, the defocusing CSP equation is derived from the defocusing CCD system based on the discussion in previous section.
\end{proof}
\subsection{Bilinear equations for the extended KP hierarchy}
Let us start with a concrete form of the Gram determinant expression of the
tau functions for the extended KP hierarchy with negative flows
\begin{equation}
\tau _{nkl}=\left\vert m_{ij}^{nkl}\right\vert _{1\leq i,j\leq N},
\end{equation}%
where
\begin{equation*}
m_{ij}^{nkl}=\delta _{ij}+\frac{1}{p_{i}+\bar{p}_{j}}\varphi _{i}^{nkl}\psi_{j}^{nkl},
\end{equation*}%
\begin{equation*}
\varphi _{i}^{nkl}=p_{i}^{n}(p_{i}-a)^{k}(p_{i}-b)^{l}e^{\xi _{i}},\quad
\psi
_{j}^{nkl}=\left(-\frac{1}{\bar{p}_{j}}\right)^{n}\left(-\frac{1}{\bar{p}_{j}+a}\right)^{k}
\left(-\frac{1}{\bar{p}_{j}+b}\right)^{l}e^{\bar{\xi}_{j}},
\end{equation*}
with
\begin{equation*}
\xi _{i}=\frac{1}{p_{i}}x_{-1}+p_{i}x_{1}+\frac{1}{p_{i}-a}t_{a}+\frac{1}{p_{i}-b}t_{b}+\xi _{i0},
\end{equation*}%
\begin{equation*}
\bar{\xi}_{j}=\frac{1}{\bar{p}_{j}}x_{-1}+\bar{p}_{j}x_{1}+\frac{1}{\bar{p}%
_{j}+a}t_{a}+\frac{1}{\bar{p}_{j}+b}t_{b}+\bar{\xi}_{j0}.
\end{equation*}%
Here $p_{i}$, $\bar{p}_{j}$, $\xi _{i0}$, $\bar{\xi}_{j0}$, $a$, $b$ are
constants. Based on the KP tau function theory~\cite{Sato,JM}, the above
tau functions satisfy a set of bilinear equations
\begin{equation}
\left(\frac{1}{2}D_{x_{1}}D_{x_{-1}}-1\right)\tau _{nkl}\cdot \tau _{nkl}=-\tau
_{n+1,kl}\tau _{n-1,kl}\,,
\label{KPbilinear1}
\end{equation}%
\begin{equation}
(aD_{t_{a}}-1)\tau _{n+1,kl}\cdot \tau _{nkl}=-\tau _{n+1,k-1,l}\tau _{n,k+1,l}\,,
\label{KPbilinear2}
\end{equation}%
\begin{equation}
\left(D_{x_{1}}(aD_{t_{a}}-1)-2a\right)\tau _{n+1,kl}\cdot \tau _{nkl}=(D_{x_{1}}-2a)\tau
_{n+1,k-1,l}\cdot \tau _{n,k+1,l}\,,
\label{KPbilinear3}
\end{equation}
\begin{equation}
\left(D_{x_{1}}(bD_{t_{b}}-1)-2b\right)\tau _{n+1,kl}\cdot \tau _{nkl}=(D_{x_{1}}-2b)\tau
_{n+1,k,l-1}\cdot \tau _{nk,l+1}\,.
\label{KPbilinear4}
\end{equation}
The proof is given below by referring to the Grammian technique~\cite{Hirotabook,MiyakeOhtaGram}.
It is easily shown that $m_{ij}^{nkl}$, $%
\varphi _{i}^{nkl}$, $\psi _{j}^{nkl}$ satisfy
\begin{equation*}
\partial _{x_{1}}m_{ij}^{nkl}=\varphi _{i}^{nkl}\psi _{j}^{nkl},\quad
\partial _{x_{-1}}m_{ij}^{nkl}=-\varphi _{i}^{n-1,kl}\psi _{j}^{n+1,kl},
\end{equation*}%
\begin{equation*}
\partial _{t_{a}}m_{ij}^{nkl}=-\varphi _{i}^{n,k-1,l}\psi _{j}^{n,k+1,l},
\end{equation*}%
\begin{equation*}
m_{ij}^{n+1,kl}=m_{ij}^{nkl}+\varphi _{i}^{nkl}\psi _{j}^{n+1,kl},\quad
m_{ij}^{n,k+1,l}=m_{ij}^{nkl}+\varphi _{i}^{nkl}\psi _{j}^{n,k+1,l}.
\end{equation*}%
Therefore the following differential and difference formulae hold for $%
\tau_{nkl}$,
\begin{equation*}
\partial_{x_{1}}\tau _{nkl}= \left\vert
\begin{matrix}
m_{ij}^{nkl} & \varphi_i^{nkl} \\
-\psi_j^{nkl} & 0%
\end{matrix}
\right\vert,\quad
\partial _{x_{-1}}\tau _{nkl}= \left\vert
\begin{matrix}
m_{ij}^{nkl} & \varphi_i^{n-1,kl} \\
\psi_j^{n+1,kl} & 0%
\end{matrix}
\right\vert,
\end{equation*}%
\begin{equation*}
a\partial _{t_{a}}\tau _{nkl}= \left\vert
\begin{matrix}
m_{ij}^{nkl} & a\varphi_i^{n,k-1,l} \\
\psi_j^{n,k+1,l} & 0%
\end{matrix}
\right\vert,\quad
\tau _{n+1,kl}= \left\vert
\begin{matrix}
m_{ij}^{nkl} & \varphi_i^{nkl} \\
-\psi_j^{n+1,kl} & 1%
\end{matrix}
\right\vert,
\end{equation*}%
\begin{equation*}
\tau _{n-1,kl}= \left\vert
\begin{matrix}
m_{ij}^{nkl} & \varphi_i^{n-1,kl} \\
\psi_j^{nkl} & 1%
\end{matrix}
\right\vert,\quad
\tau _{n,k+1,l}= \left\vert
\begin{matrix}
m_{ij}^{nkl} & \varphi_i^{nkl} \\
-\psi_j^{n,k+1,l} & 1%
\end{matrix}
\right\vert,
\end{equation*}%
\begin{equation*}
\tau _{n+1,k-1,l}=\left\vert
\begin{matrix}
m_{ij}^{nkl} & a\varphi_i^{n,k-1,l} \\
\psi_j^{n+1,kl} & 1%
\end{matrix}
\right\vert,\quad
\partial _{x_{1}}\tau _{n+1,kl}= \left\vert
\begin{matrix}
m_{ij}^{nkl} & \varphi_i^{n+1,kl} \\
-\psi_j^{n+1,kl} & 0%
\end{matrix}
\right\vert,
\end{equation*}%
\begin{equation*}
(\partial _{x_{1}}+a)\tau _{n,k+1,l}= \left\vert
\begin{matrix}
m_{ij}^{nkl} & \varphi_i^{n+1,kl} \\
-\psi_j^{n,k+1,l} & a%
\end{matrix}
\right\vert,
\end{equation*}%
\begin{equation}\label{border1}
(\partial _{x_{1}}\partial _{x_{-1}}-1)\tau _{nkl}= \left\vert
\begin{matrix}
m_{ij}^{nkl} & \varphi_i^{n-1,kl} & \varphi_i^{nkl} \\
\psi_j^{n+1,kl} & 0 & -1 \\
-\psi_j^{nkl} & -1 & 0%
\end{matrix}
\right\vert,
\end{equation}
\begin{equation}\label{border2}
(a\partial _{t_{a}}-1)\tau _{n+1,kl}= \left\vert
\begin{matrix}
m_{ij}^{nkl} & \varphi_i^{nkl} & a\varphi_i^{n,k-1,l} \\
-\psi_j^{n+1,kl} & 1 & -1 \\
\psi_j^{n,k+1,l} & -1 & 0%
\end{matrix}
\right\vert,
\end{equation}
\begin{equation}\label{border3}
(\partial _{x_{1}}(a\partial _{t_{a}}-1)-a)\tau _{n+1,kl} =\left\vert
\begin{matrix}
m_{ij}^{nkl} & \varphi_i^{n+1,kl} & a\varphi_i^{n,k-1,l} \\
-\psi_j^{n+1,kl} & 0 & -1 \\
\psi_j^{n,k+1,l} & -a & 0%
\end{matrix}
\right\vert.
\end{equation}
Applying the Jacobi identity of determinants to the bordered determinants
(\ref{border1})--(\ref{border3}), the three bilinear
equations (\ref{KPbilinear1})--(\ref{KPbilinear3}) are satisfied.
The bilinear equation (\ref{KPbilinear4}) can be proved exactly in the same way as equation (\ref{KPbilinear3}) .
\subsection{Reduction to the CCD system and the CSP equation of defocusing type}
In what follows, we briefly show the reduction processes of reducing bilinear
equations of extended KP hierarchy (\ref{KPbilinear1})--(\ref{KPbilinear4}) to the bilinear equation (\ref{CDSP_bilinear1})--(\ref{CDSP_bilinear2}).
Firstly, we start with dimension reduction by noting that the determinant expression of $\tau _{nkl}$,
\begin{equation*}
\tau _{nkl}=\left\vert \delta _{ij}+\frac{1}{p_{i}+\bar{p}_{j}}\varphi
_{i}^{nkl}\psi _{j}^{nkl}\right\vert _{1\leq i,j\leq N}\,,
\end{equation*}
can be alternatively expressed by%
\begin{equation*}
\tau _{nkl}=\left\vert \delta _{ij}+\frac{1}{p_{i}+\bar{p}_{j}}\varphi
_{i}^{nkl}\psi _{i}^{nkl}\right\vert _{1\leq i,j\leq N}\,,
\end{equation*}
by dividing $j$-th column by $\psi _{j}^{nkl}$ and multiplying $i$-th row by $%
\psi _{i}^{nkl}$ for $1\leq i,j\leq N$. By taking
\begin{equation}
\bar{p}_{j}=\frac{1}{p_{j}},\quad b=-\frac{1}{a},
\label{reduction_par}
\end{equation}%
we can easily check that
$\tau _{nkl}$ satisfies the reduction conditions
\begin{equation}
\partial _{x_{1}}\tau _{nkl}=\partial _{x_{-1}}\tau _{nkl}\,,
\label{reduction_condition1}
\end{equation}%
\begin{equation}
-a^2\partial _{t_{a}}\tau _{nkl}=\partial _{t_{b}}\tau _{nkl}\,,
\label{reduction_condition2}
\end{equation}%
\begin{equation}
\tau _{n-1,k+1,l+1}=\tau _{nkl}\,.
\label{reduction_condition3}
\end{equation}%
Therefore the bilinear equation (\ref{KPbilinear1}) is reduced to
\begin{equation}
\left(\frac{1}{2}D_{x_{1}}^{2}-1\right)\tau _{nkl}\cdot \tau _{nkl}=-\tau _{n+1,kl}\tau
_{n-1,kl}.
\label{Reduction_bilinear1}
\end{equation}
Moreover, by referring to the bilinear equation (\ref{KPbilinear4}) and the reduction conditions (\ref{reduction_condition2})--(\ref{reduction_condition3}), we have
\begin{equation*}
\left(D_{x_{1}}(aD_{t_{a}}-1)+\frac{2}{a}\right)
\tau _{n+1,kl}\cdot \tau _{nkl}
=\left(D_{x_{1}}+\frac{2}{a}\right)\tau _{n,k+1,l}\cdot \tau _{n+1,k-1,l}\,,
\end{equation*}%
thus using (\ref{KPbilinear2}) and (\ref{KPbilinear3}) we get
\begin{equation*}
\left(D_{x_{1}}(aD_{t_{a}}-1)-a+\frac{1}{a}\right)
\tau _{n+1,kl}\cdot \tau _{nkl}
=\left(-a+\frac{1}{a}\right)\tau _{n+1,k-1,l}\tau _{n,k+1,l}
\end{equation*}%
\begin{equation}
=\left(a-\frac{1}{a}\right)(aD_{t_{a}}-1)\tau _{n+1,kl}\cdot \tau _{nkl}\,,
\label{Reduction_bilinear2}
\end{equation}%
i.e.,
\begin{equation}
(D_{x_{1}}(aD_{t_{a}}-1)-(a^2-1)D_{t_{a}})\tau _{n+1,kl}\cdot
\tau _{nkl}=0\,.
\label{Reduction_bilinear3}
\end{equation}
Next, we proceed to the reduction of complex conjugate, which turns out to be very simple. Specifically, by taking $a$ pure imaginary,
$|p_{i}|=1$ and $\bar{\xi}_{j0}=\xi _{j0}^{\ast}$,
where $^{\ast}$ means complex conjugate, we have $\bar{p}_{i}=p^{\ast}_{i}$ and
\begin{equation*}
\tau _{n00}^{\ast }=\tau _{-n,00}\,.
\end{equation*}
Due to the relation (\ref{reduction_condition1}) and (\ref{reduction_condition2}), we can choose 
$x_{1}$ (or $x_{1}+x_{-1}$) and $t_{a}$ (or $t_{a}+t_{b}$ ) as two independent variables. Therefore, we have
\begin{equation*}
\tau _{n00}=\left\vert \delta _{ij}+\frac{1}{p_{i}+p_{j}^{\ast }}\left(-\frac{p_{i}%
}{p_{j}^{\ast }}\right)^{n}e^{\xi _{i}+\xi _{j}^{\ast }}\right\vert _{1\leq
i,j\leq N}\,,
\end{equation*}
with
\begin{equation*}
\xi _{i}=p_{i}x_{1}+\frac{1}{p_{i}-a}t_{a}+\xi _{i0}\,.
\end{equation*}
In summary, by defining
\begin{equation*}
f=\tau _{000},\quad g=\tau _{100},
\end{equation*}%
we arrive at
\begin{equation}
\left(D_{x_{1}}D_{t_{a}}-\frac{1}{a}D_{x_{1}}
-\left(a-\frac{1}{a}\right)D_{t_{a}}\right)g\cdot f=0\,,
\label{Reduction_bilinear4}
\end{equation}
\begin{equation}
\left(\frac{1}{2}D_{x_{1}}^{2}-1\right)f\cdot f=-gg^{\ast}\,.
\label{Reduction_bilinear5}
\end{equation}
Finally, by setting $a=\mathrm{i} c$, $t_{a}= cy$, $x_{1}= \beta s/4$ and $\beta (c^2+1)=-2 \gamma c$, the above bilinear equations
coincide with the bilinear equations (\ref{CDSP_bilinear1})--(\ref{CDSP_bilinear2}). Therefore, the reduction process is complete. As a result, we can provide the determinant solution to the defocusing CSP equation by the following theorem.
\begin{theorem}
The defocusing CSP equation (\ref{dCSP}) admits the following determinant solution
\begin{equation}
q= \frac{\beta}{2} \frac{g}{f}e^{\mathrm{i}( y+\gamma s/2)},\quad x= -\frac{\gamma}{2}  y+
\frac{\beta^2}{8}s-2(\log f)_{s}\,\quad t=-s,
\label{ndark-qhodo}
\end{equation}
where%
\begin{eqnarray}
f &=&\left\vert \delta _{ij}+\frac{1}{p_{i}+p_{j}^{\ast }}e^{\xi _{i}+\xi
_{j}^{\ast }}\right\vert _{1\leq i,j\leq N}\,, \nonumber \\
g &=&\left\vert \delta _{ij}+\left( -\frac{p_{i}}{p_{j}^{\ast }}\right)
\frac{1}{p_{i}+p_{j}^{\ast }}e^{\xi _{i}+\xi _{j}^{\ast }}\right\vert
_{1\leq i,j\leq N}\,,
\label{ndark-tau}
\end{eqnarray}
with
\begin{equation}
|p_{i}|=1,\quad \xi_{i}=\frac{c }{p_{i}-\mathrm{i}c}y+\frac{\beta}{4} p_{i} s+\xi _{i0}\,,
\label{ndark-par}
\end{equation}
under a constraint
\begin{equation}
\frac{\beta}{\gamma} = -\frac{c^2+1}{2 c}\,.
\label{ndark-constraint}
\end{equation}
\end{theorem}
In the last, we list one- and two-dark soliton solutions to the defocusing CSP equation (\ref{dCSP}).
By taking $p_{1}=e^{-\mathrm{i}(\varphi_{1}+\pi/2)}$ and $N=1$ in (\ref{ndark-tau}), we have the tau functions
for one-dark soliton solution,
\begin{equation}
f=1+e^{2\eta _{1}},\quad g=1+e^{2(\eta _{1}-\mathrm{i}\varphi _{1})}\,,
\end{equation}
where
\begin{equation*}
2\eta _{1}=- \frac{ \beta \sin \varphi _{1}}{2} \left( s + \frac {2y}{\beta \cos \varphi _{1} -\gamma} \right)-\ln (p_{1}+p_{1}^{\ast}) \,.
\end{equation*}
This leads to a one-dark soliton solution of the following parametric form
\begin{equation}  \label{CSP1solitona}
q= \frac{\beta}{2}\left((1+\mathrm{e}^{-2\mathrm{i}\varphi_{1}}) + (\mathrm{e}^{-2\mathrm{i}\varphi _{1}}-1)
\tanh \eta_{1}  \right)\mathrm{e}^{\mathrm{i}( y +\gamma s/2)}\,,
\end{equation}
\begin{equation}  \label{CCD1solitonb}
x=  -\frac{\gamma}{2}y+\frac{\beta^2}{8}s+\frac{\beta  \sin \varphi_{1}
e^{2\eta _{1}}}{1+e^{2\eta _{1}}},\quad \,t=-s\,.
\end{equation}
\begin{figure}[htb]
\centering
\includegraphics[height=60mm,width=80mm]{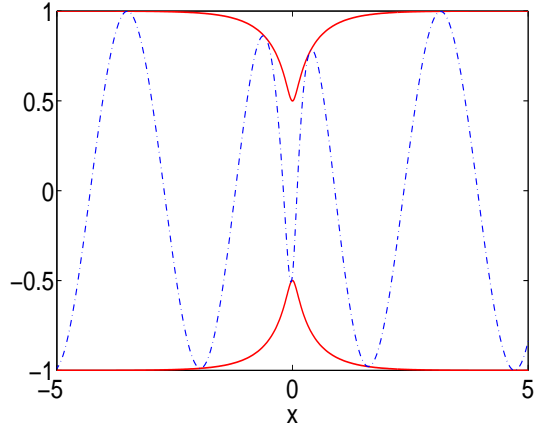}
\caption{A smoothed dark soliton with $\beta=2.0$, $\gamma=1.0$, $\varphi_1=2\pi/3$.}
\label{fig2}
\end{figure}
Obviously, the amplitude of background plane waves is $\beta/2$, the
depth of the trough is $\beta (1-|\cos \varphi_1|)/2$.

An example with $\beta=2.0$, $\gamma=1.0$, $\varphi_1=2\pi/3$ is
illustrated in Fig. 1. In this case, the envelope of the dark soliton is
smooth. However, as analysed in \cite{FenglingzhuPRE}, if we take
$\beta=(2+2 \sqrt{7})/3$ while keeping other parameters unchanged, the
dark soliton becomes
a cusped envelope soliton,
as shown in Fig. 2. In other words, the dark solution has more tendency
to become singular (
cusped
or looped one) as the amplitude of the background waves increases.
\begin{figure}[htb]
\centering
\includegraphics[height=60mm,width=80mm]{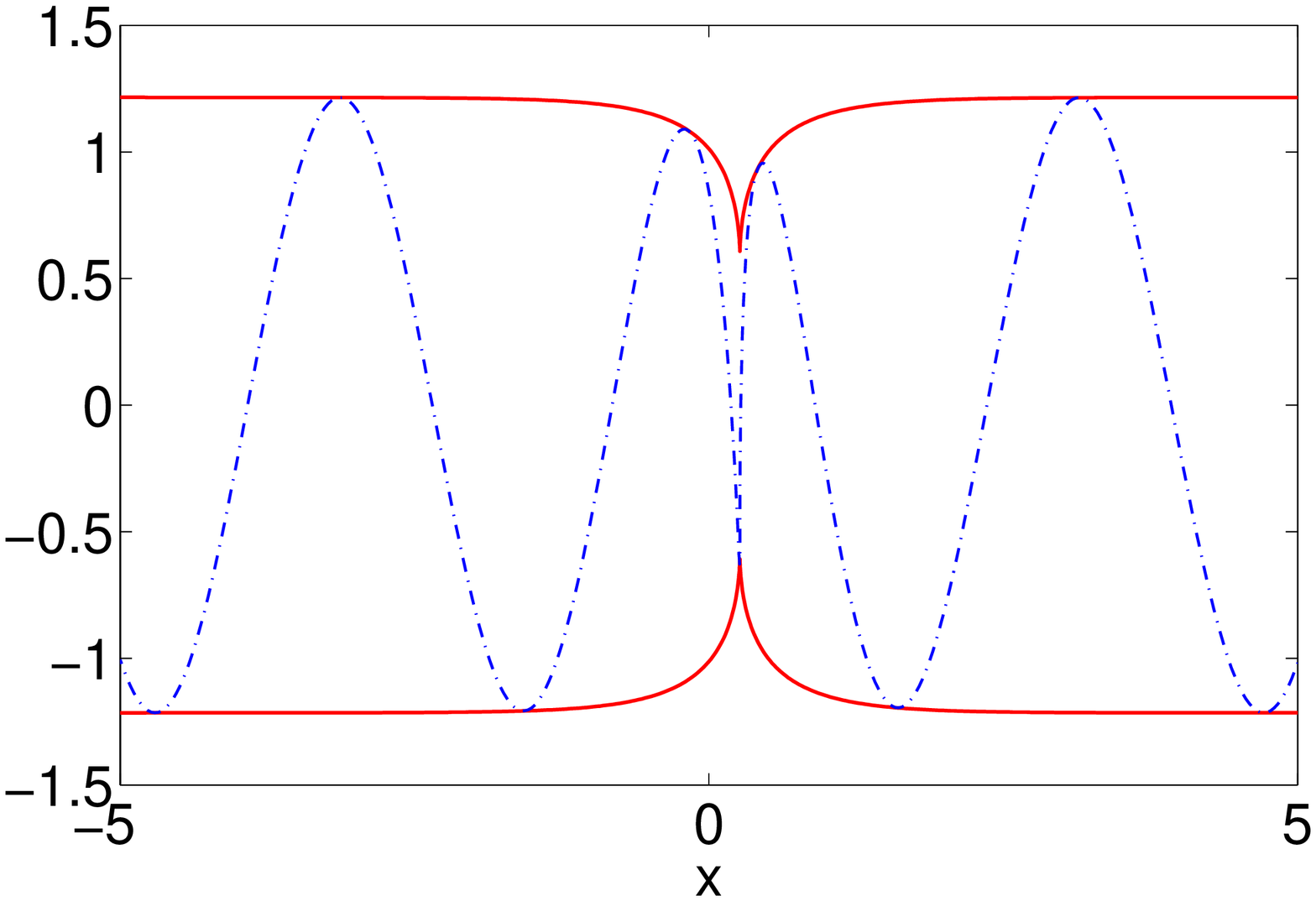}
 \caption{A
 cusped
 dark soliton with $\beta=(2+2\sqrt{7})/3$, $\gamma=1.0$, $\varphi_1=2\pi/3$.}
\label{fig3}
\end{figure}
From (\ref{ndark-tau}) with $N=2$, we obtain the tau functions for
two-soliton solution, 
\begin{eqnarray}
f &=&\left\vert
\begin{array}{cc}
1+\frac{1}{p_{1}+p_{1}^{\ast }{}}e^{\xi _{1}+\xi _{1}^{\ast }} & \frac{1}{%
p_{1}+p_{2}^{\ast }{}}e^{\xi _{1}+\xi _{2}^{\ast }} \\
\frac{1}{p_{2}+p_{1}^{\ast }{}}e^{\xi _{2}+\xi _{1}^{\ast }} & 1+\frac{1}{%
p_{2}+p_{2}^{\ast }{}}e^{\xi _{2}+\xi _{2}^{\ast }}%
\end{array}%
\right\vert \nonumber\\
&=&1+e^{2\eta _{1}}+e^{2\eta _{2}}+a_{12}e^{2(\eta _{1}+\eta _{2})},
\label{2solitong}
\end{eqnarray}

\begin{eqnarray}
g &=&\left\vert
\begin{array}{cc}
1+\frac{1}{p_{1}+p_{1}^{\ast }{}}(-\frac{p_{1}}{p_{1}^{\ast }})e^{\xi
_{1}+\xi _{1}^{\ast }} & \frac{1}{p_{1}+p_{2}^{\ast }{}}(-\frac{p_{1}}{%
p_{2}^{\ast }})e^{\xi _{1}+\xi _{2}^{\ast }} \\
\frac{1}{p_{2}+p_{1}^{\ast }{}}(-\frac{p_{2}}{p_{1}^{\ast }})e^{\xi _{2}+\xi
_{1}^{\ast }} & 1+\frac{1}{p_{2}+p_{2}^{\ast }{}}(-\frac{p_{2}}{p_{2}^{\ast }%
})e^{\xi _{2}+\xi _{2}^{\ast }}%
\end{array}%
\right\vert \nonumber\\
&=&1+e^{2(\eta _{1}-\mathrm{i}\varphi _{1})}+e^{2(\eta _{2}-\mathrm{i}%
\varphi _{2})}+a_{12}e^{2(\eta _{1}+\eta _{2}-\mathrm{i}\varphi _{1}-\mathrm{%
i}\varphi _{2})},
\end{eqnarray}%
where
\begin{equation*}
2\eta _{j}=- \frac{ \beta \sin \varphi _{1}}{2} \left( s + \frac {2y}{\beta \cos \varphi _{j} -\gamma} \right)-\ln (p_{j}+p_{j}^{\ast})\,,\quad j=1,2\,,
\end{equation*}
\begin{equation*}
a_{12} =\frac{\sin ^{2}\left( \frac{\varphi
_{2}-\varphi _{1}}{2}\right) }{\sin ^{2}\left( \frac{\varphi _{2}+\varphi
_{1}}{2}\right) }\,.
\end{equation*}
The collision of two-dark solitons to the defocusing CSP equation is always elastic, whose analysis is given in \cite{FenglingzhuPRE}.

 \section{Concluding Remarks}
In \cite{FenglingzhuPRE}, a defocusing CSP equation was derived from
physical context in nonlinear optics as an analogue of the NLS equation in ultra-short pulse regime.
In the present paper, we have established a link between the defocusing
CSP equation and the motion of space curves
in Minkowski space and the complex sinh-Gordon equation by a hodograph (reciprocal) transformation.
We have also derived the CSP equation of both focusing and defocusing type from the fundamental forms of surfaces with non-constant Gaussian curvature, from which the curve flows are formulated.
Secondly, starting from a set of bilinear equations, along with their tau functions,
of a single-component extended KP hierarchy, we have derived the defocusing CSP equation
based on the KP-hierarchy reduction method. Meanwhile, its multi-soliton
solutions have been provided in determinant form.

Even though we have recently constructed various solutions including
bright soliton, dark soliton, breather and rogue wave solutions
to the the focusing and defocusing CSP
equation~\cite{Feng_ComplexSPE,FengShen_ComplexSPE,FLZPhysD,FenglingzhuPRE}
including the work in the present paper,
it will be more interesting to investigate all kinds of
solutions to the coupled CSP equation,
especially the one of mixed focusing and defocusing nonlinearity.
In the last, although integrable discretizations of the real short pulse
equation and its multi-component generalizations were recently
constructed~\cite{FMO-PJMI,FengSPEdiscrete,Fengplanecurves,FCCMOcSP,FMOmultiSP},
how to understand and reconstruct the integrable discretizations of the
focusing and defocusing CSP equation geometrically and algebraically
and relate them with a general frame work set in
~\cite{BobenkoSurisbook,SchiefdiscretePLR}
remains a topic to be explored in the future.

\section*{Acknowledment}
We thank for the reviewers' comments which helped us to improve the
manuscript significantly.
BF appreciate the useful discussion on geometric part of the paper with Dr. Zhiwei Wu.
BF appreciated the partial support by the National Natural Science
Foundation of China (No. 11428102).
The work of KM is partially supported by JSPS Grant-in-Aid for Scientific Research (C-15K04909)
and CREST, JST.
The work of YO is partly supported by JSPS Grant-in-Aid for Scientific
Research (B-24340029, C-15K04909) and
for Challenging Exploratory Research (26610029).

\end{document}